\documentclass[11pt, a4paper]{article}

\usepackage{amsmath,amsfonts,amsthm,amssymb}
\usepackage[retainorgcmds]{IEEEtrantools}
\usepackage[ruled,vlined,linesnumbered]{algorithm2e}
\usepackage{tikz}
\usepackage{kerkis}
\usepackage{fourier}
\providecommand{\DontPrintSemicolon}{\dontprintsemicolon}
\usepackage{hyperref}
\usepackage[normalem]{ulem}
\usepackage{changepage}
\usepackage{microtype}
\usepackage[normalem]{ulem}
\usepackage{graphicx} 
\usepackage{paralist}

\newcommand{\myheader}[1]{\bigskip\noindent\textbf{#1.}~~}

\theoremstyle{plain}
\newtheorem{theorem}{Theorem}[section]

\newtheorem{lemma}[theorem]{Lemma}
\newtheorem{proposition}[theorem]{Proposition}
\newtheorem{claim}[theorem]{Claim}
\newtheorem{definition}[theorem]{Definition}
\newtheorem{corollary}[theorem]{Corollary}

\theoremstyle{remark}

\newtheorem{remark}[theorem]{Remark}

\renewcommand{\vec}[1]{\mathbf{#1}}

\makeatletter
\newcommand*{\algotitle}[2]{%
	\stepcounter{algocf}%
	\hypertarget{algocf.title.\theHalgocf}{}%
	\NR@gettitle{#1}%
	\label{#2}%
	\addtocounter{algocf}{-1}%
}
\makeatother

\DeclareMathOperator*{\argmin}{arg\,min}

\newcommand{\M}{\text{IACSM}}
\newcommand{\A}{\ensuremath{A}}
\newcommand{\T}{\ensuremath{T}}
\newcommand{\mtag}{\ensuremath{\textsf{tag}}}
\newcommand{\set}[1]{\ensuremath{\{#1 \}}}
\newcommand{\sset}[2]{\ensuremath{\{#1 \, : \, #2 \}}}
\newcommand{\vneg}{\vspace*{-\medskipamount}}

\newcommand{\dv}{\ensuremath{v}}

\allowdisplaybreaks

\usepackage[a4paper, hmargin=2.5cm,top=2.5cm,bottom=2.5cm,nohead]{geometry}

\title{\bfseries Cost Sharing over Combinatorial Domains: \\ Complement-Free Cost Functions and Beyond\thanks{A conference version appears in the proceedings of the 27th Annual European Symposium on Algorithms, ESA 2019.}}

\author{
	Georgios Birmpas$^1$\hspace{10pt} 
	Evangelos Markakis$^2$\hspace{10pt} 
	Guido Sch\"afer$^3$  \vspace{5pt}
	\\ 
        $^{1}$ \normalsize{Department of Computer Science, University of Oxford, UK}\\
        $^{2}$ \normalsize{Department of Informatics, Athens University of Economics and Business, Greece}\\
	$^3$ \normalsize{Networks and Optimization Group, Centrum Wiskunde \& Informatica (CWI), the Netherlands} \vspace{2pt} \\
	\normalsize{\texttt{gebirbas@gmail.com\hspace{10pt}
	markakis@aueb.gr\hspace{10pt}
	g.schaefer@cwi.nl}}
}

\begin{document}
\thispagestyle{empty}
\newpage

\maketitle

\sloppy
\begin{abstract}
We study mechanism design for combinatorial cost sharing models. Imagine that multiple items or services are available to be shared among a set of interested agents. The outcome of a mechanism in this setting consists of an assignment, determining for each item the set of players who are granted service, together with respective payments.  Although there are several works studying specialized versions of such problems, there has been almost no progress for general combinatorial cost sharing domains until recently \cite{DobzinskiO17}.  Still, many questions about the interplay between strategyproofness, cost recovery and economic efficiency remain unanswered.

The main goal of our work is to further understand this interplay in terms of budget balance and social cost approximation. 
Towards this, we provide a refinement of cross-monotonicity (which we term \emph{trace-monotonicity}) that is applicable to iterative mechanisms. The trace here refers to the order in which players become finalized. On top of this, we also provide two parameterizations (complementary to a certain extent) of cost functions which capture the behavior of their average cost-shares.

Based on our trace-monotonicity property, we design a scheme of ascending cost sharing mechanisms which is applicable to the combinatorial cost sharing setting with symmetric submodular valuations. 
Using our first cost function parameterization, we identify conditions under which our mechanism is weakly group-strategyproof, $O(1)$-budget-balanced and $O(H_n)$-approximate with respect to the social cost. Further, we show that our mechanism is budget-balanced and $H_n$-approximate if both the valuations and the cost functions are symmetric submodular; given existing impossibility results, this is best possible. 

Finally, we consider general valuation functions and exploit our second parameterization to derive a more fine-grained analysis of the Sequential Mechanism introduced by Moulin. This mechanism is budget balanced by construction, but in general only guarantees a poor social cost approximation of $n$. We identify conditions under which the mechanism achieves improved social cost approximation guarantees. In particular, we derive improved mechanisms for fundamental cost sharing problems, including Vertex Cover and Set Cover. 
\end{abstract}

\fussy

\pagestyle{plain}

\section{Introduction}
\label{sec:intro}



How to share the cost of a common service (or public good) among a set of interested agents constitutes a fundamental problem in mechanism design that has been studied intensively for at least two decades. Several deep and significant advancements have been achieved throughout this period, notably also combining classical mechanism design objectives (such as incentive compatibility, economic efficiency, etc.) with theoretical computer science objectives (such as approximability and computational efficiency). 

However, in the vast majority of the cost sharing models that have been proposed and analyzed in the literature, it is assumed that the mechanism designer is offering a single service and that each agent has a private value describing the willingness to pay for the service. At the same time, there is also a publicly known cost function which describes the total cost for offering the service to each possible subset of agents. Said differently, this results in a single-parameter mechanism design problem, where the goal is to select a subset of the players that will be granted service, subject to covering the cost and achieving an economically efficient outcome.


Although significant progress has been made for such single-parameter domains, moving towards more general \emph{combinatorial domains} has been almost elusive so far. Imagine that there are multiple goods to be shared among the agents who now have more complex valuation functions, expressing their willingness to pay for different subsets (or bundles) of goods. The cost function now depends on the subsets of agents sharing each of the items. An outcome of a mechanism under this setting, consists of an allocation, which specifies for each agent the goods for which she is granted service, together with a payment scheme.

The desired properties in designing a cost-sharing mechanism (be it combinatorial or not) are three-fold: 
(i) \emph{group-strategyproofness}: we would like resistance to misreporting preferences by individual agents or coalitions, 
(ii) \emph{budget-balance}: the payments of the players should cover the incurred cost, 
(iii) \emph{economic efficiency}: the allocation should maximize a measure of social efficiency. 
The fundamental results in \cite{GKL76,Roberts79} rule out the possibility that all three properties can be achieved. As a result, if we insist on any variant of strategyproofness, we are forced to settle with approximate notions of at least one of the other two criteria. In this context, approximate budget balance means that the mechanism may overcharge the agents, but not by too much. In terms of efficiency, considering a social cost objective instead of the classical social welfare objective (definitions are given in Section \ref{sec:defs}) seems more amenable for multiplicative approximation guarantees.

These adapted objectives have been investigated thoroughly for single-parameter problems, especially for cost-sharing variants of well-known optimization problems. 
In the context of more general combinatorial cost-sharing mechanisms, a restricted model with multiple levels of service was first studied in \cite{MehtaRS09}. Ever since, for almost a decade, there was no additional progress along these lines. It was only recently that a step forward was made by Dobzinski and Ovadia \cite{DobzinskiO17}. In their work, they introduce a combinatorial cost-sharing model and derive the first mechanisms guaranteeing good budget balance and social cost approximation guarantees for different classes of valuation and cost functions. As already pointed out in \cite{DobzinskiO17}, however, several important questions concerning our understanding of the approximability of these objectives remain open and deserve further study. This constitutes the starting point of our investigations reported in this work. 

\myheader{Our Contributions}
We make further advancements on the design and analysis of mechanisms for combinatorial cost-sharing models. 
To begin with, the analysis of the mechanisms we study asks for new conceptual ideas (which might be interesting on their own): 
\begin{itemize}
	\item We first provide a refinement of the well-known notion of \emph{cross-monotonic} cost sharing functions, which is key in the intensively studied class of Moulin-Shenker mechanisms \cite{MS01} for the single-parameter domain. We introduce the notion of \emph{trace-monotonic} cost sharing functions which is applicable for mechanisms that proceed iteratively and evict agents one-by-one. Trace-monotonicity formalizes the fact that the cost-shares observed by a player for an item do not decrease throughout the \emph{course} of the mechanism. That is, these cost shares may depend on the specific order (or \emph{trace} as we will call it) in which the mechanism considers the agents. 
	\item We identify two different and (to some extent) complementary parameterizations of the cost functions. Intuitively, these parameters measure the ``variance'' of the average cost-share $c(S)/|S|$, over all agent sets $S$. We introduce two such notions, which we term \emph{$\alpha$-average decreasing} and \emph{$\alpha$-average min-bounded} (see Definition \ref{def:avg-dec} and Definition \ref{def:cmin}, respectively). We note that for every cost function, there exist respective values of $\alpha$ (possibly different for each definition) for which these properties are satisfied. These definitions provide an alternative way to classify cost functions and their respective approximation guarantees in terms of budget balance and social cost. 
\end{itemize}

Using the above ideas, in Section \ref{sec:mech}, we derive a scheme for ascending cost sharing mechanisms, which can be seen as a (non-trivial) adaptation of the Moulin-Shenker mechanisms from the binary accept/reject setting to combinatorial cost sharing. Our notion of trace-monotonic cost shares plays a crucial role here. We show that our proposed mechanism is applicable for any non-decreasing cost function and for symmetric submodular valuations (i.e., submodular functions whose value depends only on the cardinality of the set). 

By exploiting the first parameterization of $\alpha$-average decreasing cost functions, our main result of Section \ref{sec:mech} is that for $\alpha = O(1)$, our mechanism is polynomial-time, weakly group-strategyproof, $O(1)$-budget-balanced and $O(H_n)$-approximate with respect to social cost, where $n$ is the number of agents \footnote{We  use $H_n$ to denote the \emph{$n$-th Harmonic number} defined as $H_n = 1 + \frac{1}{2} + \dots + \frac{1}{n}$.}. As a consequence, if both the valuation and cost functions are symmetric submodular ($\alpha=1$), the mechanism is budget-balanced and $H_n$-approximate. This is best possible even for a single item, as there exist corresponding inapproximability results by Dobzinski et al.~\cite{DobzinskiMRS18}. 
Prior to our work, the best known mechanism for symmetric submodular valuation and cost functions is $H_n$-budget balanced and $H_n$-approximate \cite{DobzinskiO17}. We anticipate that further extensions and generalizations might be feasible through our framework and this type of ascending mechanisms.

In Section \ref{sec:FPM}, we exploit our second parameterization of $\alpha$-average min-bounded cost functions, and provide results for general valuation functions. As it turns out, our parameterization enables us to obtain a more fine-grained analysis of the Sequential Mechanism introduced by Moulin \cite{Moulin99}. 
This mechanism is budget-balanced by construction, but in general only guarantees a poor social cost approximation of factor $n$. 
We show that for $\alpha$-average min-bounded cost functions with $\alpha = O(1)$, the Sequential Mechanism is budget balanced and $H_n$-approximate with respect to social cost. 
Interestingly, this result does not even require monotonicity of the valuation functions. In addition, we can push our results even a bit further by introducing a refinement of this class of cost functions (see Definition~\ref{def:cmax}) for which we show that the Sequential Mechanism is $O(1)$-approximate. 
The refinement allows us to obtain improved mechanisms for several cost functions originating from combinatorial optimization problems. For example, our result implies that the Sequential Mechanism is $d$-approximate for certain cost-sharing variants of Vertex Cover and Set Cover, where $d$ is the maximum degree of a node or the maximum size of a set, respectively; this improves upon existing results, even in the well-studied single-item case, when $d$ is constant. 

In general, the two parameterizations of the cost functions introduced in this work seem to be suitable means to accurately capture the approximation guarantees of both the ascending cost-sharing mechanism of Section \ref{sec:mech}, and the Sequential Mechanism of Section \ref{sec:FPM}. In fact, we have not managed to construct natural examples of cost functions which do not admit an $O(H_n)$-approximation by neither of the mechanisms studied here. See also the discussion in Section \ref{sec:dis}, where some examples are constructed but they are rather artificial (Proposition \ref{prop:intersec}). 
As such, these parameterizations help us to narrow down the class of cost functions which are not yet known to admit a good social cost approximation and enhance our understanding towards further progress in combinatorial cost sharing. 

\myheader{Related Work} For the single-item setting and with submodular cost functions, the best known group-strategyproof and budget balanced cost-sharing mechanism is arguably the Shapley value mechanism, introduced by Moulin and Shenker \cite{Moulin99,MS01}. 
This was also the first work that tried to quantify the efficiency loss of budget balanced cost-sharing mechanisms. 
Later, Feigenbaum et al.~\cite{FeigenbaumKSS03} showed that if one insists on truthfulness, there is no mechanism that achieves a finite approximation of the social welfare objective, even if one relaxes the budget balance property to cost recovery. 
To overcome this impossibility result, Roughgarden et al.~\cite{RoughgardenS09} introduced the notion of \emph{social cost} as an alternative means to quantify the efficiency of a mechanism. In the same work, they showed that the Shapley value mechanism is $H_n$-approximate with respect to this objective. 
Dobzinski et al.~\cite{DobzinskiMRS18} established another impossibility result for the social cost objective, and showed that every mechanism satisfying truthfulness and cost recovery cannot achieve a social cost approximation guarantee better than $\Omega(\log{n})$. 
The problem of deriving mechanisms with the best possible budget balance and social cost approximation guarantees for different cost functions arising from combinatorial optimization problems has been extensively studied in various works, see e.g., \cite{BleischwitzS08,BrennerS07,ChawlaRS06,GuptaKLRS15}.

Moving beyond the single-item case, Mehta et al.~\cite{MehtaRS09} introduced a new family of truthful mechanisms (called \emph{acyclic mechanisms}) which apply to general demand settings of multiple identical items when players have symmetric submodular valuations. For additional works that consider the general demand setting, the reader is referred to~ \cite{BleischwitzS08,BrennerS07,DevanurMV05,Moulin99}.  
Birmpas et al.~\cite{BirmpasC0M15} also studied families of valuation and cost functions for the multiple item setting, under cost sharing models that are motivated by applications in participatory sensing environments. 

Most related to our work is the recent work by Dobzinski and Ovadia~\cite{DobzinskiO17}. To the best of our knowledge, this is the only prior work that considers a more general approach for combinatorial cost sharing. 
They studied a multi-parameter setting and proposed a new VCG-based mechanism. Basically, the idea is to run a VCG mechanism~\cite{Clarke71,Groves73,Vickrey61} with respect to a modified objective function which is defined as the sum of the player valuations minus a potential. 
Intuitively, the latter ensures that the  payments computed by the mechanism cover the actual cost. 
They showed that this mechanism is strategyproof and $H_n$-approximate with respect to social cost. They also identified several classes of valuation and cost functions for which the mechanism is $H_n$-budget balanced. In particular, this is the case if the valuation and cost functions are symmetric.\footnote{We note that their definition of symmetry for the cost function differs from the one we use here.}
Additionally, they established that their mechanism is optimal with respect to the social cost approximation among all symmetric VCG-based mechanisms that always cover the cost.

\section{Definitions and Notation}
\label{sec:defs}

We assume there is a set $N = \{1, 2, \dots, n\}$ of players and a set $M=\{1,2,\dots,m\}$ of items. Each item can be viewed as a public good 
or some service that can be shared by the players. Each player $i$ has a \emph{private valuation function $v_i:2^{M}\to \mathbb{R}_{\geq 0}$} specifying the value that she derives from each subset of items.  

A \emph{cost-sharing mechanism} takes as input the declared (possibly false) valuation functions $\vec{b} = (b_i)_{i \in N}$ of the players and outputs (i) an allocation that determines which players share each item and (ii) a payment $p_i$ for each player $i$.
An allocation is denoted by a tuple $A = (A_1, \dots, A_n)$, where $A_i \subseteq M$ is the set of items provided to player $i$. For notational convenience, we also represent an allocation $A = (A_i)_{i \in N}$ as a tuple over the items space $(T_1, \dots, T_m)$ such that for every item $j \in M$, $T_j \subseteq N$ is the subset of players sharing item $j$, i.e., $T_j = \sset{i\in N}{j\in A_i}$.

In this paper, we consider mostly \emph{separable} cost functions. In the \emph{separable setting}, we assume that the overall cost of an allocation decomposes into the cost for providing each item separately. Hence, every item $j$ is associated to a known cost function $c_j:2^{N}\to \mathbb{R}_{\geq 0}$, which specifies for each set of players $T \subseteq N$, the cost $c_j(T)$ of providing item $j$ to the players in $T$. 
Thus, the total cost of an allocation $A$ is defined as 
\begin{equation}
\label{eq:cost-def}
C(A) = \sum_{j\in M} c_j(T_j)
\end{equation}
In Section \ref{subsec:nonsep}, we also consider the \emph{non-separable setting}, where we are given a more general cost function $C: (2^M)^n \rightarrow \mathbb{R}_{\ge 0}$, specifying for every allocation $A = (A_i)_{i \in N}$ the corresponding cost $C(A)$. Non-separable functions can capture dependencies among different items. 

We assume that the utility functions of the players are \emph{quasilinear}, i.e., given an allocation $A = (A_i)_{i \in N}$ and payments $(p_i)_{i \in N}$ determined by the mechanism for valuation functions $\vec{\dv} = (\dv_i)_{i \in N}$, the utility of player $i$ is defined as $u_i(\vec{v}) = v_i(A_i) - p_i$.
All our mechanisms have \emph{no positive transfers} (NPT), i.e., $p_i\geq 0$, and satisfy \emph{individual rationality} (IR), i.e., $p_i \le \dv_i(A_i)$.

In addition to the above, we are also interested in the following properties:
\begin{itemize}
\item \textbf{Weak Group-Strategyproofness (WGSP):}
	We insist on a stronger notion of resistance to manipulation than truthfulness: A mechanism is \emph{weakly group-strategyproof} if there is no deviation by a coalition of players that makes all its members strictly better off. More formally, we require that for every coalition $Q \subseteq N$ of players, every profile $\vec{v}_{-Q}$ of the other players, there is no deviation $\vec{b}_Q$ of the players in $Q$ such that $u_i(\vec{b}_Q, \vec{v}_{-Q}) > u_i(\vec{v}_Q, \vec{v}_{-Q})$ for every $i \in Q$, where $\vec{v}_Q$ is the profile of the actual valuation functions of $Q$.
	

\item \textbf{Budget Balance:}
	We are interested in mechanisms whose payments cover the allocation cost, ideally exactly. However, the latter is not always possible as it may be incompatible with the other objectives. We therefore consider an approximate budget balance notion: A mechanism is \emph{$\beta$-budget-balanced} ($\beta\geq 1$) if for every valuation profile $\vec{\dv} = (\dv_i)_{i \in N}$, the outcome $(A, p)$ computed by the mechanism satisfies
	\begin{equation}
	C(A) \leq \sum_{i\in N} p_i \leq \beta\cdot C(A).
	\end{equation}
	Clearly, we want $\beta$ to be as small as possible to not overcharge players too much for covering the cost. We say that the mechanism is \emph{budget balanced} if $\beta = 1$.
	
\item \textbf{Economic Efficiency:} 
	Our goal is to compute outcomes that are (approximately) efficient. To this aim, we use the \emph{social cost objective}, originally defined in \cite{RoughgardenS09}. Adapted to our combinatorial setting, the \emph{social cost} of an allocation $A = (A_i)_{i \in N}$ is defined as the actual cost of the outcome plus the value missed by not serving all items to all players, i.e.,
	\begin{equation}
	\pi(A) = \sum_{j\in M} c_j(T_j) + \sum_{i\in N}[ \dv_i(M)-\dv_i(S_i)].\footnote{Note that this adaptation was proposed in \cite{DobzinskiO17}. }
	\end{equation}
	A mechanism is said to be \emph{$\alpha$-approximate} with respect to the social cost objective if for every valuation profile $\vec{\dv} = (\dv_i)_{i \in N}$, the allocation $A$ output by the mechanism satisfies 
	\begin{equation}
	\pi(A) \leq \alpha \cdot \pi(A^*),
	\end{equation}
	where $A^*$ is an allocation of minimum social cost.
	\end{itemize}

We assume that both the valuation functions $(v_i)_{i \in N}$ and the cost functions $(c_j)_{j \in M}$ are non-decreasing (see below for formal definitions). 
Further, we focus on certain classes of valuation and cost functions: More specifically, we consider \emph{submodular} and \emph{subadditive} cost functions, both naturally modeling economies of scale. As to the valuation functions, we consider \emph{submodular} valuation functions in Section~\ref{sec:mech} and general valuation functions in Section~\ref{sec:FPM}. 
Further, the class of \emph{symmetric XOS} functions play a prominent role in Section~\ref{sec:mech}.\footnote{It is not hard to verify that these functions can equivalently be defined as stated in Definition~\ref{def:sym-XOS} (see also \cite{EzraFRS18}).}
Below we summarize all relevant definitions (see also Lehman et al. \cite{LehmanLN06}).

\begin{definition}\label{def:sym-XOS}
	Let $f: 2^U \rightarrow \mathbb{R}_{\ge 0}$ be a function defined over subsets of a universe $U$.
	\begin{compactenum}
		\item {$f$ is \emph{non-decreasing} if $f(S) \le f(T)$ for every $S \subseteq T \subseteq U$.}
		\item $f$ is \emph{symmetric} if $f(S)=f(T)$ for every $S, T \subseteq U$ with $|S| = |T|$. 
		\item $f$ is \emph{submodular} if $f(S \cup \{i\}) - f(S) \geq f(T\cup \{i\}) - f(T)$ for every $S\subseteq T \subseteq U$ and $i\not\in S$.
		\item {$f$ is \emph{XOS} if there are additive functions $f^1, \dots, f^k$ such that $f(S) = \max_{i \in [k]} f^i(S)$ for all $S \subseteq U$.}
		\item  $f$ is \emph{subadditive} if $f(S \cup T) \leq f(S) + f(T)$ for every $S, T \subseteq U$.
		
		\item $f$ is \emph{symmetric XOS} if it is symmetric and $f(S)/|S| \ge f(T)/|T|$ for every $S \subseteq T \subseteq U$.
	\end{compactenum}
\end{definition}

Some of our mechanisms make use of cross-monotonic cost-sharing functions defined as follows:
\begin{definition}
	Let $c: 2^N \rightarrow \mathbb{R}_{\ge 0}$ be a cost function. A \emph{cost-sharing function}\footnote{We stress here that it is also possible for cost-sharing methods to overcharge, something that leads to approximate budget balanced mechanisms.} $\chi: N \times 2^N \rightarrow \mathbb{R}_{\ge 0}$ with respect to $c$ specifies for each subset $S \subseteq N$ and every player $i \in S$ a non-negative cost share $\chi(i, S)$ such that $\sum_{i\in S} \chi(i, S) = c(S)$.\footnote{For notational convenience, we define $\chi(i, S) = \infty$ for $i \notin S$.}
	$\chi$ is \emph{cross-monotonic} if for all $S \subseteq T \subseteq N$ and every $i \in S$, we have $\chi(i, S) \ge \chi(i, T)$.
\end{definition}

\section{An Iterative Ascending Cost Sharing Mechanism}
\label{sec:mech}

In this section, we present our \emph{Iterative Ascending Cost Sharing Mechanism ($\M$)} for the combinatorial cost sharing setting with symmetric submodular valuations and general cost functions. We first provide a generic description of our mechanism and identify two properties which are sufficient for our main result to go through. We then show that these properties are satisfied if the valuations are symmetric submodular. 

\subsection{Definition of $\M$}
\label{subsec:def-iacsm}
\label{subsec:twoprop}

\begin{algorithm}[t]
	\small
	\DontPrintSemicolon 
	\caption{Iterative Ascending Cost Sharing Mechanism ($\M$)} \label{fig:alg-1} 
	
	\KwIn{Declared valuation functions $(b_i)_{i \in N}$.}
	\KwOut{Allocation $A = (A_i)_{i \in N}$ and payments $p = (p_i)_{i \in N}$.}
	
	\textbf{Initialization:} Let $X = N$ be the set of active players and define $T_j = N$ for every item $j \in M$. 
	\label{alg:s0}
	
	\While{$X \neq \emptyset$}{
		Compute an \emph{optimal bundle} $A_i$ for every player $i \in X$: 
		\vneg
		\label{alg:s1}
		\begin{align}\label{eq:opt-choice}
		A_i & \in \arg\max_{S \subseteq M} \{b_i(S) - p_i(S)\}, \quad \text{where} \quad p_i(S) = \sum_{j \in S} \chi_j(T_j) \\[-5ex] \notag
		\end{align}
		(If there are several optimal bundles, resolve ties as described within Section \ref{subsec:def-iacsm}.) \;

		Let $i^* \in X$ be a player such that $|A_{i^*}| \le |A_i|$ for every $i \in X$.\;
		\label{alg:s2}
		
		Assign the items in $A_{i^*}$ to player $i^*$ and remove player $i^*$ from $X$. \;
		\label{alg:s3}
		
		For every item $j \in M \setminus A_{i^*}$, set $T_j = T_j \setminus \set{i^*}$, and update the cost shares $\chi_j(T_j)$.\;
		\label{alg:s4}
		
	}
	\Return $A = (A_i)_{i \in N}$ and $p = (p_i)_{i \in N}$, where $p_i = \sum_{j \in A_i} \chi_j(T_j)$.
	\label{alg:s5}
\end{algorithm}

Mechanism $\M$ can be viewed as a generalization of the \emph{Moulin-Shenker} mechanism \cite{MS01} to the combinatorial setting in the sense that it simulates in parallel an ascending iterative auction for each item. To our knowledge this is the first ascending price mechanism for the combinatorial setting which is not VCG-based and as we will describe below, this adaptation is not straightforward since there are several obstacles we need to overcome. A description of our mechanism $\M$ is given in Algorithm~\ref{fig:alg-1}. 

The mechanism maintains a set of \emph{active} players $X$ and for each item $j \in M$ a set of players $T_j$ who are \emph{tentatively} assigned to $j$. Initially, each player is active and tentatively assigned to all the items, i.e., $X = N$ and $T_j = N$ for all $j \in M$. 
The mechanism then proceeds in iterations. 
In each iteration, each item $j$ is offered to each active player $i \in X$ at a price that only depends on the set of tentatively assigned players $T_j$. For this, we use a player-independent \emph{cost sharing function} $\chi_j(\cdot, T_j)$ for every item $j$, and since we require that $\chi_j(i, T_j) = \chi_j(k, T_j)$ for every $i, k\in T_j$, we will simply denote by $\chi_j(T_j)$ the cost share that each player $i \in T_j$ tentatively assigned to $j$ has to pay.  
Based on these cost shares, every active player $i \in X$ computes an \emph{optimal bundle} $A_i$ with respect to the payments $p_i(\cdot)$, as defined in Equation \eqref{eq:opt-choice}. If there are ties, we resolve them according to the following \emph{tie-breaking rule}: if there are several optimal bundles, then player $i$ chooses one of maximum size. 
If there are multiple optimal bundles of maximum size $k$, then she chooses the bundle consisting of the $k$ cheapest items (where ties between equal cost share items are resolved consistently, say by index of the items). 

After determining the optimal bundle for each active player, the mechanism then chooses an active player $i^*$ whose optimal bundle has minimum size. Again, we break ties consistently, say by index of the players. 
The items in $A_{i^*}$ are assigned to player $i^*$ and $i^*$ becomes inactive. Finally, for every item $j$ which is not part of the optimal bundle $A_{i^*}$, $i^*$ is removed from the tentative set $T_j$. The mechanism terminates when all players are inactive.

\subsection{Two Crucial Properties}
\label{subsec:twoprop}

In this section we identify two crucial properties that our mechanism has to satisfy for our main result to go through. To formalize these properties, we introduce first some more notation. 

\myheader{Trace of $\M$}
Note that the execution of our mechanism $\M$ on an instance of the problem induces an order $\tau = (\tau_1, \dots, \tau_n)$ on the players.
Without loss of generality, we may assume that the players are renamed such that $\tau = (1, \dots, n)$, i.e., player $i$ becomes inactive in iteration $i$; however, we emphasize that this order is determined by the run of our mechanism. 

The order $\tau = (1, \dots, n)$ together with the final bundle $A_i$ assigned to each player $i$ at the end of iteration $i$ induces an order of player withdrawals for each item $j$. More precisely, for every $j\in M$ we let $\tau_j$ be the subsequence of $\tau$ consisting only of the players who withdrew from item $j$ (at the end of the iteration when they became inactive). We refer to $\tau_j$ as the \emph{trace of item $j$}. Recall that initially $T_j = N$ and hence all players are tentatively assigned to $j$. The length of $\tau_j$ can vary from $0$, when nobody withdraws from item $j$ and $\tau_j$ is the null sequence, all the way to $n$, when everybody withdraws from $j$ and $\tau_j = \tau$.   
Given a trace $\tau_j$ in the form $\tau_j= ({i_1}, {i_2}, \dots, {i_\ell})$ and $k\in \{0, 1, \dots, |\tau_j|\}$, let $R_j^k = N\setminus \{{i_1}, {i_2}, \dots , {i_k}\}$; define $R_j^0 = N$. Note that the set $R_j^k$ is precisely the set of players tentatively assigned to $j$ after $k$ players have withdrawn from $j$ during the execution of the mechanism. 
We note that the notion of trace is valid also for any other iterative mechanism where the assignment of one player becomes finalized at each iteration, e.g., \cite{MehtaRS09}.

\myheader{Trace-monotonic cost sharing functions} 
We introduce a new property of cost sharing functions which will turn out to be crucial below. Intuitively, it is a refinement of the standard cross-monotonicity property which has to hold only for \emph{certain} subsets of players encountered by the mechanism, namely for the sets $\set{R^k_j}$. More precisely, 
given a trace $\tau_j$ for an item $j\in M$, we say that the cost sharing function $\chi_j$ is \emph{cross-monotonic with respect to $\tau_j$} (or, \emph{trace-monotonic} for short), if 
\[\forall k \in \set{0, \dots, |\tau_j|-1}: \qquad \chi_j(R^k_{j}) \le \chi_j(R^{k+1}_{j}).
\]
Note that this ensures that the cost share of item $j$ (weakly) increases during the execution of the mechanism, as we consider the sequence of sets 
\[
R_j^0 \supset R_j^1 \supset \dots \supset R_j^{|\tau_j|}.
\]
A subtle point here is that the definition of the cost share $\chi_j(R_j^k)$ may not only depend on the set of players $R^k_j$, but also on the trace $\tau_j$ specifying how the set $R^k_j$ has been reached by the mechanism.\footnote{Notationally, we would have to write here $\chi^{\tau_j}_j$ to indicate the dependency on $\tau_j$. However, in the analysis we focus on a fixed trace produced by an execution of the mechanism and omit the explicit reference to it for notational convenience.} 
It will become clear below that this additional flexibility enables us to implement our mechanism for \emph{arbitrary} cost functions.

\myheader{Properties (P1) and (P2)}
Our first property is rather intuitive: An item $j$ needs to be offered to all active players at the same price and this price can only increase in subsequent iterations. In particular, this ensures that if at the end of iteration $i$, player $i$ withdraws from an item $j \in M \setminus A_i$, then the price of $j$ for the remaining players in $T_j \setminus \set{i}$ does not decrease. This is crucial to achieve strategyproofness, and it is captured precisely by trace-monotonic cost sharing functions.\footnote{Note that we have to require trace-monotonicity with respect to an \emph{arbitrary} trace of item $j$ here, because we cannot control the trace $\tau_j$ that will be realized by $\M$. 
}
%
\begin{enumerate}[\textbf{(P1)}]
	\item  For each item $j \in M$ the cost sharing function $\chi_j$ is trace monotonic for every trace $\tau_j$.
\end{enumerate}

The first property alone is not sufficient to ensure that our mechanism $\M$ is weakly group-strategyproof (or even strategyproof). 
Additionally, we need to enforce the following refinement property on the final bundles assigned to the players. We prove below that Property (P2) is satisfied for symmetric submodular valuation functions.
\begin{enumerate}[\textbf{(P2)}]\setcounter{enumi}{1}
	\item The final bundles $(A_i)_{i \in N}$ assigned to the players satisfy the following \emph{refinement property}: $A_i \subseteq A_{i+1}$ for every $i \in \set{1, \dots, n-1}$.
\end{enumerate}

\subsection{Feasibility of (P1) and (P2)}

We next define the cost sharing function that we use. The intuition is as follows: Suppose that $S = T_j$ is the set of players who are tentatively allocated to item $j$ at the beginning of iteration $i$ for some $i\in [n]$. Ideally, we would like to charge the average cost $c_j(S)/|S|$ to each player in $S$, but we cannot simply do this because the average cost might decrease with respect to the previous iteration, and this will destroy Property (P1). Given our new notion of trace-monotonicity, we can resolve this by defining the cost share of item $j$ as the maximum average cost over all player sets which were tentatively allocated to $j$ so far. 

More formally, let $\tau_j$ be the trace of item $j$ induced by $\M$ when executed on a given instance. Let $S$ be the set of players tentatively assigned to item $j$ at the beginning of iteration $i$, and fix $k$ such that $R_j^k = S$ (by the definition of our mechanism, such a $k$ must exist and $k \le i-1$). 
We define 
\begin{equation}\label{eq:max-avg-csr}
\chi_j(S) = \max_{\ell \in \set{0, \dots, k}}  
\frac{c_j(R_j^\ell)}{|n-\ell|}.
\end{equation}

Note that by using this definition we may end up overcharging the actual cost $c_j(S)$ of item $j$ in the sense that $|S| \cdot \chi_j(S) > c_j(S)$. As we show in Section \ref{subsec:main}, the budget balance and social cost approximation guarantees depend on the magnitude by which we might overcharge.

It is now trivial to show that Property (P1) holds. 

\begin{lemma}\label{lem:average-cs}
	Consider some item $j \in M$ and let $c_j: 2^N \rightarrow \mathbb{R}_{\ge 0}$ be an arbitrary cost function. Let $\tau_j$ be an arbitrary trace of $j$.
	The cost sharing function $\chi_j$ defined in \eqref{eq:max-avg-csr} is trace-monotonic.
\end{lemma}

\begin{proof}
By definition \eqref{eq:max-avg-csr}, we have for every $k \in \set{0, \dots, |\tau_j|-1}$ 
	$$
	\chi_j(R^k_j) 
	= \max_{\ell \in \set{0, \dots, k}} \frac{c_j(R_j^\ell)}{|n-\ell|} 
	\le \max_{\ell \in \set{0, \dots, k+1}} \frac{c_j(R^\ell_j)}{|n-\ell|} 
	= \chi_j(R_j^{k+1}).
	$$
\end{proof}

We turn to Property (P2). In general, it seems difficult to guarantee (P2), but it is not hard to see that it holds if the valuation functions are symmetric submodular. 

\begin{lemma}\label{lem:refine}
	Suppose the valuation functions are symmetric submodular. Then $A_i \subseteq A_{i+1}$ for every $i \in \set{1, \dots, n-1}$.
\end{lemma}

\begin{proof}
Fix some $i \in \set{1, \dots, n-1}$ and consider players $i$ and $i+1$. 
	Note that both $i$ and $i+1$ are active at the beginning of iteration $i$. 
	Let $A_i$ and $A'_{i+1}$ be the optimal bundles chosen by $i$ and $i+1$ in iteration $i$, respectively. 
	Because $i$ is chosen, we have $|A_i| \le |A'_{i+1}|$. 
	Further, because the valuation functions are symmetric, $A_i$ consists of the $|A_i|$ smallest cost share items (by our tie-breaking rule). Similarly, $A'_{i+1}$ consists of the $|A'_{i+1}|$ smallest cost share items. We conclude that $A_i \subseteq A'_{i+1}$. (Note that here we exploit that if there are several optimal bundles for player $i+1$, then the one of maximum size is chosen.) 
	
	Note that at the end of iteration $i$, player $i$ becomes inactive and withdraws from the items in $A'_{i+1} \setminus A_i$. By trace-monotonicity (Property (P1)), the cost shares of these items do not decrease in iteration $i+1$. Also, the cost shares of all items in $A_i$ remain the same.
	
	Consider now player $i+1$ in iteration $i+1$. By using similar arguments, it follows that the optimal bundle $A_{i+1}$ consists of the $|A_{i+1}|$ lowest cost share items. 
	But note that in iteration $i$, the optimal bundle $A'_{i+1}$ of player $i+1$ contained all the items of $A_i$ and possibly a few more. Since the items of $A_i$ continue to be the ones of smallest cost share, the optimal bundle of $i+1$ in iteration $i+1$ must still contain all the items in $A_i$. (Note that here we again exploit that the optimal bundle of maximum size is chosen if there are ties.) Thus it must hold that $A_{i+1} \supseteq A_i$.
\end{proof}

\subsection{Main result for IACSM}
\label{subsec:main}

In order to state our main result of this section, we need to introduce a crucial parameter that determines the budget balance and social cost approximation guarantees of our mechanism. 

\begin{definition} \label{def:avg-dec}
	A cost function $c: 2^N \rightarrow \mathbb{R}_{\ge 0}$ is \emph{$\alpha$-average-decreasing} if there exists some $\alpha \ge 1$ such that for every $S \subseteq T \subseteq N$, $\alpha \cdot \frac{c(S)}{|S|} \ge \frac{c(T)}{|T|}$. 
\end{definition}

Note that for every cost function $c$ there exists some $\alpha \ge 1$ such that $c$ is $\alpha$-average decreasing. However, here we are particularly interested in $\alpha$-average decreasing cost functions for which the parameter $\alpha$ is small, as can be seen by Theorem \ref{thm:main} below. 
Average decreasing functions with small values of $\alpha$ arise naturally in the domains of digital goods and public goods. For digital goods the cost of serving a non-empty set of customers is typically assumed to be constant because there is a cost for producing the good and then it can be shared with no additional cost (hence the definition is satisfied with $\alpha=1$). The same is applicable for some public good models. Note also that \emph{symmetric XOS cost functions} (see Definition~\ref{def:sym-XOS}) are average-decreasing (i.e., $\alpha = 1$).

The following is the main result of this section.

\begin{theorem}\label{thm:main}
	Suppose the valuation functions are symmetric submodular and the cost functions are $\alpha$-average decreasing. Then the mechanism $\M$ runs in polynomial time, satisfies IR, NPT, WGSP and is $\alpha$-budget balanced and $2 \alpha^3 H_n$-approximate.
\end{theorem}

Symmetric submodular cost functions are average decreasing (i.e., $\alpha = 1$) since they are a subclass of symmetric XOS functions. As a consequence, we obtain the following corollary from Theorem~\ref{thm:main} (with an additional improvement on the social cost approximation). 

\begin{corollary}\label{col:adm}
	Suppose the valuation functions and the cost functions are symmetric submodular. Then the mechanism $\M$ runs in polynomial time, satisfies IR, NPT, WGSP and is budget balanced and $H_n$-approximate.
\end{corollary}

\begin{proof}
Applying Theorem \ref{thm:main} with $\alpha=1$, we get immediately all the claimed properties except for the social cost approximation which is $2H_n$.
	However, we note that for symmetric submodular cost functions the term $2\alpha$ in Lemma~\ref{lem:adiffer-2} (see Section \ref{subsec:sum}) can simply be omitted, because we have $c(S) \leq c(T)$ for any two sets $S, T$ with $|S|\leq |T|$. 
	By exploiting this in the remaining proof of the social cost approximation guarantee, we obtain an $H_n$-approximation. 
\end{proof}

Note that the approximation factor of $H_n$ for symmetric submodular functions is tight: The impossibility result of Dobzinski et al.~\cite{DobzinskiMRS18} for a single public good implies that achieving a better approximation ratio is impossible, even in the single-item case ($m = 1$).


Finally we point out that \emph{$\alpha$-average-decreasing} functions are subadditive when $\alpha=1$, while this is not necessarily true for $\alpha>1$.

\begin{lemma}\label{lem:subadditivity}  
	Let $c(\cdot)$ be an $\alpha$-average-decreasing cost function where $\alpha=1$. Then $c(\cdot)$ is subadditive and in addition, not necessarily symmetric, or submodular. In case $c(\cdot)$ is $\alpha$-average-decreasing with $\alpha>1$, then $c(\cdot)$ is not necessarily subadditive.
\end{lemma}

\begin{proof}
Let $S, T \subseteq N$ and assume without loss of generality that 
	$\frac{c(S)}{|S|} \leq \frac{c(T)}{|T|}$. 
	Using that the average cost of $c$ is non-increasing, we obtain 
	\begin{align*}
	c(S\cup T) 
	& 
	\leq |S\cup T|\cdot \frac{c(S)}{|S|} 
	\leq c(S) + |T| \frac{c(S)}{|S|} \leq c(S)+c(T).
	\end{align*}
	Hence $c$ is subadditive.
	
	We construct now an example of a non-symmetric and non-submodular cost function that is average-decreasing. Let $N = \set{1, 2, 3}$ and consider the function $c$ defined as follows: $c(\{1\})=5$, $c(\{2\})=7$, $c(\{3\})=8$, $c(\{1,2\})=10$, $c(\{2,3\}) = c(\{1,3\})=9$, and $c(\{1,2, 3\}) = 11$. It is easy to verify that $c$ has decreasing average cost, but obviously, it is not a symmetric function; in fact, it is not even submodular because $c(\{1,2, 3 \})-c(\{1, 3\})=11-9=2>1=9-8=c(\{2, 3\})-c(\{3\})$.     
	
	Finally consider the following example:
\[	c(S) = 
	\begin{cases}
	0 & \text{if $S=\emptyset$} \\
	1 & \text{if $|S|=1,2$} \\
	3 & \text{if $|S|\geq 3$} 
	\end{cases} 
\]
	
	It is easy to see that this function is    \emph{2-average-decreasing} (just consider sets of cardinality 2 and supersets of cardinality 3 and notice that this case gives the maximum possible $\alpha$). Regarding non-subadditivity, consider sets $S\subseteq T \subseteq N$ where $|S|=1, |T|=2$ and $S\cap T=\emptyset$. We have that $c(S\cup T)=3>2=1+1=c(S)+c(T)$.
\end{proof}

The remainder of this section is devoted to the proof of Theorem~\ref{thm:main}. 
Unless stated otherwise, we assume below that the valuation functions are symmetric submodular and the cost functions are $\alpha$-average decreasing.

\subsection{Computational Efficiency}

We argue that $\M$ can be implemented in polynomial time. Clearly, the mechanism terminates after $n$ iterations. 
In each iteration, the only non-trivial operations are (i) to compute the optimal bundles for all active players (Line 3) and (ii) to update the cost shares of the items (Line 6). All other operations can be implemented to run in time $O(n+m)$. We prove below that (i) and (ii) can be done in polynomial time. 

We first show that optimal bundles can be computed efficiently if the valuation functions are symmetric submodular.  

\begin{lemma}\label{cl:utility-max}
	If the valuation functions are symmetric submodular, then an optimal bundle for player $i$ can be computed in polynomial time.
\end{lemma}
\begin{proof}
	Fix an arbitrary iteration and let $i \in X$ be an active player. We need to show that we can efficiently compute an optimal bundle as defined in \eqref{eq:opt-choice}. 
	Recall that the value that player $i$ derives from a bundle of items only depends on its cardinality (because the valuation functions are symmetric submodular). Thus, to build an optimal bundle, we can start with the empty set and iteratively add an item of lowest price as long as this price is at most the added marginal value. As a result, the value of the constructed bundle (weakly) increases whenever we add a new item. At the same time, the marginal value of the added items (weakly) decreases because of submodularity. Thus, the first time we consider an item whose marginal value is (strictly) less than its price, we have to stop as the utility can only decrease if we add any of the remaining items. Note that the optimal bundle constructed in this way satisfies the tie-breaking rule described above. Clearly, this procedure stops after at most $m$ iterations. 
\end{proof}

We note that for non-symmetric valuation functions, the complexity of $\M$ depends on the time needed to compute an optimal bundle. For example, if one has access to demand queries, we can still have an efficient implementation with polynomially many queries.

We next turn to the computation of the cost shares. Note that in general it is not clear whether the cost shares as defined in \eqref{eq:max-avg-csr} can be computed efficiently (as there might be exponentially many supersets that need to be considered). However, as we show below, the cost shares of the items which are constructed throughout the execution of our mechanism $\M$ can be computed efficiently. 

\begin{lemma}\label{lem:cost-share-comp}
	The cost shares as defined in \eqref{eq:max-avg-csr} can be computed efficiently throughout the execution of mechanism $\M$.
\end{lemma}
\begin{proof}
	Let $\tau = (1, 2, \dots, n)$ be the player order induced by $\M$. 
	Fix an arbitrary item $j \in M$. After the initialization, we have $T_j = N$ and the cost share of $j$ is thus $\chi_j(R_j^0) = c_j(N)/|N|$. Clearly, as long as the set of players $T_j$ who are tentatively assigned to item $j$ does not change, the cost share of $j$ remains the same. Suppose that $T_j$ changes at the end of iteration $i$ because player $i$ withdraws from item $j$, resulting in a new set $T'_j = T_j \setminus \set{i}$. The new cost share $\chi_j(T'_j)$ of item $j$ can then be determined simply by taking the maximum of the current cost share $\chi_j(T_j)$ and $c_j(T'_j)/|T'_j|$. Note that this update ensures that the cost share definition in \eqref{eq:max-avg-csr} is met. 
\end{proof}


\subsection{IR, NPT and WGSP}
\label{subsec:prop}

The individual rationality and the no positive transfers properties follow directly from the definition of the mechanism.  The WGSP property is established by the following lemma.


\begin{lemma}\label{lem:tr}
	Mechanism $\M$ is weakly group-strategyproof.
\end{lemma}
\begin{proof}    
	Fix a coalition $Q \subseteq N$. Let $I$ be the instance in which all players in $Q$ report their valuations truthfully and let $\hat{I}$ be an instance in which all players in $Q$ misreport their valuations. We need to prove that not every player $i \in Q$ strictly improves her utility by misreporting. 
	
	Consider the runs of $\M$ on instances $I$ and $\hat{I}$, respectively. 
	Let $\tau$ and $\hat{\tau}$ be the player orders induced by $\M$ for $I$ and $\hat{I}$, respectively.
	We assume without loss of generality that in the run on $I$, player $i$ is considered in iteration $i$, i.e., ${\tau} = (1, \dots, n)$.
	\footnote{Note that the mechanism might terminate before iteration $n$, but for the analysis it will be convenient to assume that it uses exactly $n$ iterations. Conceptually, simply assume that the players who become inactive all at once in the final iteration are removed one-by-one (using an arbitrary but consistent tie breaking rule).} 
	Let $\hat{\tau} = (\hat{\tau}_1, \dots, \hat{\tau}_n)$ be the order in which the players are considered in the run on $\hat{I}$. 
	Let $i$ be the first iteration in which either (a) the considered players $i$ and $\hat{\tau}_i$ differ, i.e., $i \neq \hat{\tau}_i$, or (b) the same player $i = \hat{\tau}_i$ is considered, but the bundles allocated in $I$ and $\hat{I}$ differ. 
	Note that such an $i$ must exist as otherwise both runs return the same allocation and we are done. By the choice of $i$, at the beginning of iteration $i$ the cost share of each item is the same in $I$ and $\hat{I}$, and can only increase subsequently. We distinguish two cases: 
	
	\emph{Case 1: $i \in Q$.} The bundle $A_i$ allocated to player $i$ in $I$ is chosen such that 
	\begin{equation}\label{eq:max-bundle}
	\dv_i(A_i) - p_i(A_i) \ge \dv_i(S) - p_i(S) \qquad \forall S \subseteq M, 
	\end{equation}
	where $p_i(S)$ is the sum of the cost shares of the items in $S$ in iteration $i$ in $I$ (which are the same as in $\hat{I}$). 
	Suppose player $i$ is considered in iteration $k$ in the run on $\hat{I}$, i.e., $\hat{\tau}_k = i$. By the choice of $i$, we have $k \ge i$. Thus, in iteration $k$ in $\hat{I}$ the price $\hat{p}_i(S)$ of $i$ for each bundle $S \subseteq M$ must satisfy $\hat{p}_i(S) \ge p_i(S)$. In particular, for the bundle $\hat{A}_i$ allocated to player $i$ in iteration $k$ in $\hat{I}$ we have 
	$$
	\dv_i(\hat{A}_i) - \hat{p}_i(\hat{A}_i) \le 
	\dv_i(\hat{A}_i) - p_i(\hat{A}_i) \le 
	\dv_i(A_i) - p_i(A_i), 
	$$
	where the last inequality follows from \eqref{eq:max-bundle}. 
	The claim follows because $i$ is part of the deviating coalition $Q$.

	\emph{Case 2: $i \notin Q$.} 
	We first argue that player $k = \hat{\tau}_i \in Q$. Assume that $k \notin Q$. Then both $i$ and $k$ bid truthfully in iteration $i$ in $I$ and $\hat{I}$. Because in iteration $i$ the cost shares are the same in both runs, the bundles chosen by $i$ in $I$ and $\hat{I}$ are the same. The same holds for player $k$. If $i \neq k$ then this is a contradiction to the assumption that $\M$ uses a consistent tie breaking rule (as $i$ is chosen in $I$ but $k$ in $\hat{I}$). If $i = k$ then this is a contradiction to our choice of iteration $i$ (as $i$ chooses the same bundle in $I$ and $\hat{I}$). We conclude that $k \in Q$. 
	
	We now compare the utility obtained by player $k = \hat{\tau}_i \in Q$ in $I$ and $\hat{I}$. 
	Note that $k \neq i$ because $k \in Q$ and $i \notin Q$. Observe that in iteration $i$ in instance $\hat{I}$, player $i$ reports truthfully and thus opts for the same bundle $A_i$ as in iteration $i$ in $I$. Given that player $k$ is chosen in iteration $i$ in $\hat{I}$ (and not player $i$), the bundle $\hat{A}_k$ allocated to $k$ satisfies $\hat{A}_k  \subseteq A_i$. 
	
	Now, consider the run on $I$ and let $k > i$ be the iteration in which player $k = \hat{\tau}_i$ is considered. By Property (P2), we  have $A_i \subseteq A_k$. We conclude that $\hat{A}_k \subseteq A_k$. 
	Because $k$ reports truthfully in $I$, the choice of $A_k$ implies that
	$$
	\dv_k(A_k) - p_k(A_k) \ge \dv_k(S) - p_k(S) \qquad \forall S \subseteq M.
	$$
	In particular, for the bundle $\hat{A}_k$ this implies that
	\begin{equation}\label{eq:pos-add}
	\dv_k(A_k) - \dv_k(\hat{A}_k) - p_k(A_k \setminus \hat{A}_k) \ge 0.
	\end{equation}
	Further, note that the cost shares of all items in $A_i$ and $\hat{A}_k$ remain the same as in iteration $i$ in the runs on $I$ and $\hat{I}$, respectively. Exploiting  that $\hat{A}_k \subseteq A_i$ and that the cost shares in iteration $i$ are the same in both runs, we conclude that $p_k(\hat{A}_k) = \hat{p}_k(\hat{A}_k)$. 
	We obtain
	\begin{align*}
	\dv_k(A_k) - p_k(A_k) 
	& =\dv_k(\hat{A}_k) - p_k(\hat{A}_k) + [\dv_k(A_k) - \dv_k(\hat{A}_k) - p_k(A_k \setminus \hat{A}_k)] \\
	& \ge \dv_k(\hat{A}_k) - {p}_k(\hat{A}_k) =\dv_k(\hat{A}_k) - \hat{p}_k(\hat{A}_k),
	\end{align*}
	where the first equality holds because $p_k$ is additive and the inequality follows from \eqref{eq:pos-add}. 
	The claim now follows because $k \in Q$.
\end{proof}

\subsection{Budget Balance and Social Cost Approximation}
\label{subsec:sum}

We start by providing the budget balance performance of the mechanism. 
\begin{lemma}\label{lem:bb} 
	Mechanism $\M$ is $\alpha$-budget balanced.
\end{lemma}
\begin{proof}
	Let the player order induced by $\M$ be $\tau = (1, \dots, n)$. Fix an arbitrary item $j$ and let $\tau_j$ be the trace of item $j$. Let $T_j$ be the final set of players allocated to item $j$ and recall that $R_j^k = T_j$ for $k = |\tau_j|$. 
	By the definition of the cost sharing function $\chi_j$ in \eqref{eq:max-avg-csr}, there is some set $S = R^{\ell}_j$ with $\ell \in \set{0, \dots, k}$ such that the cost share of item $j$ is
	\begin{equation}\label{eq:alpha-bb}
	\chi_j(T_j) = \frac{c_j(R^\ell_{j})}{n-\ell} \ge \frac{c_j(R^k_{j})}{n-k} = \frac{c_j(T_j)}{|T_j|}.
	\end{equation}
	Summing over all players in $T_j$, we obtain 
	\begin{equation}\label{eq:alpha-p}
	c_j(T_j) 
	\leq \sum_{i\in T_j}\chi_j(T_j)=|T_j|\cdot \frac{c_j(S)}{|S|} \leq |T_j|\cdot \alpha \frac{c_j(T_j)}{|T_j|}={\alpha} \cdot c_j(T_j),
	\end{equation}
	where the second inequality holds because $c_j$ is $\alpha$-average-decreasing and $T_j \subseteq S$. 
	
	Finally, summing inequality \eqref{eq:alpha-p} over all items $j \in M$ we obtain 
	\begin{equation}\label{eq:alpha-tbb}
	\sum_{j \in M}c_j(T_j) 
	\leq \sum_{j \in M}\sum_{i\in T_j}\chi_j(T_j) 
	\leq \alpha\sum_{j \in M}c_j(T_j), 
	\end{equation}
	which proves the claim.
\end{proof}

\bigskip\noindent
We now show that our mechanism $\M$ is $2\alpha^3 H_n$-approximate with respect to the social cost objective for symmetric submodular valuation functions. 

Let $\A = (\A_i)_{i \in N}$ be the allocation computed by the mechanism, where $\A_i \subseteq M$ is the subset of items that player $i$ receives. As before, without loss of generality we assume that the player order induced by $\M$ is $\tau = (1, \dots, n)$. 
Recall that for every item $j \in M$, $\T_j = \sset{i \in N}{j \in \A_i}$ is the final set of players that receive item $j$. 
We also use $\T^i_j$ to refer to the subset of players who are allocated to item $j$ at the beginning of iteration $i$. Clearly, $\T^i_j \supseteq \T_j$ for every player $i$ and item $j$.


%

We first prove some lemmas which will be helpful later on. 

\begin{lemma}\label{lem:final}
	Fix an item $j \in M$ and let $i$ be the first player in $\tau$ such that $j \in \A_i$. Then $\T_j = \set{i, \dots, n}$.
\end{lemma}
\begin{proof}
	By the choice of $i$, we have that $j \in \A_i$ and $j \notin \A_k$ for every player $k < i$. 
	From Property (P2) it follows that for every $k$ with $i \le k < n$, $\A_{k} \subseteq \A_{k+1}$. Thus $j \in \A_{k}$ for every $i \le k \le n$, which concludes the proof.	
\end{proof}

\begin{lemma}\label{lem:opt-bundle}
	Consider player $i$ who becomes inactive in iteration $i$. 
	We have 
	$$
	\dv_i(\A_i) - \sum_{j \in \A_i} \chi_j(T_j) \ge 
	\dv_i(S) - \sum_{j \in S} \chi_j(T_j) \quad \forall S \subseteq M.
	$$
\end{lemma}
\begin{proof}
	In iteration $i$, the final bundle $\A_i$ is chosen as the set of items maximizing the utility of player $i$ with respect to the current cost shares, i.e., 
	\begin{equation}\label{eq:local-shares}
	\dv_i(\A_i) - \sum_{j \in \A_i} \chi_j(\T^i_j) \ge 
	\dv_i(S) - \sum_{j \in S} \chi_j(\T^i_j) \quad \forall S \subseteq M.
	\end{equation}
	Recall that $\T^i_j$ is the set of players that are allocated to item $j$ in iteration $i$.
	Note that by Lemma~\ref{lem:final}, $\T^i_j = \T_j$ for every $j \in \A_i$. Further, $\T^i_j \supseteq \T_j$ for every $j \in M \setminus \A_i$ as additional players might withdraw from $j$ in subsequent iterations. Note that the final set $T_j$ is reached from $T_j^i$ by following the trace $\tau_j$ of item $j$. The claim now follows from the trace-monotonicity of $\chi_j$ (Property (P1)). 
\end{proof}

\begin{lemma}\label{lem:acost-share-bound}
	Consider player $i$ who becomes inactive in iteration $i$. 
	For every item $j \in M$, 
	$$
	\chi_j(\T^i_j) \le {\alpha}\frac{c_j(\set{i, \dots, n})}{n-i+1}.
	$$
\end{lemma}
\begin{proof}
	Fix some $j \in M$. 
	In iteration $i$, we have $\T^i_j \supseteq \set{i, \dots, n}$. 
	Let $k$ be such that $\T^i_j = R^k_j$ and recall that $R_j^\ell \supseteq R_j^k$ for every $\ell \in \set{0, \dots, k}$. 
	We obtain 
	\begin{equation}\label{eq:acost-share-bound}
	\chi_j(\T^i_j) = \max_{\ell \in \set{0, \dots, k}} \frac{c_j(R^\ell_j)}{n-\ell} \le  {\alpha}\frac{c_j(\set{i, \dots, n})}{n-i+1},
	\end{equation}
	where the inequality holds because $c_j$ is $\alpha$-average decreasing.
\end{proof}

\begin{lemma}\label{lem:adiffer-2}
	Let $c$ be an $\alpha$-average decreasing cost function.
	Let $S, T \subseteq N$ be arbitrary subsets with $|S| \le |T|$. 
	Then $c(S) \le {2}{\alpha}c(T)$. 
\end{lemma}

\begin{proof}
	Assume for the sake of a contradiction that $c(S) > {2}{\alpha}c(T)$. Consider the set $S \cup T$ and note that $|S \cup T| \le 2|T|$. Because $c$ is non-decreasing, we have $c(S \cup T) \ge c(S) >\frac{2}{\alpha} c(T)$. Using that $c$ is $\alpha$-average-decreasing, we obtain $\frac{c(S \cup T)}{|S \cup T|} \le {\alpha}\frac{c(T)}{|T|}$, which implies that $c(S \cup T) \le {\alpha}\frac{|S \cup T|}{|T|} c(T) \le {2}{\alpha} c(T)$, a contradiction.	
\end{proof}



We are now ready to prove the approximation guarantee. 

\begin{lemma}
	Mechanisms $\M$ is $2\alpha^3 H_n$-approximate.
\end{lemma}
\begin{proof}
	Let $\A^* = (\A^*_1, \dots, \A^*_n)$ be an optimal allocation and let $\T^*_j$ be the respective set of players that receive item $j$ in $\A^*$. 
	We have 
	\begin{align*}
	\pi(\A) 
	& = \sum_{i \in N} \big(\dv_i(M) - \dv_i(\A_i)\big) + \sum_{j \in M} c_j(\T_j) \\
	& \leq \sum_{i \in N} \dv_i(M) - \sum_{i \in N} \bigg(\dv_i(\A_i) - \sum_{j \in \A_i} \chi_j(\T_j) \bigg) \\
	& \le \sum_{i \in N} \dv_i(M) - \sum_{i \in N} \bigg(\dv_i(\A^*_i) - \sum_{j \in \A^*_i} \chi_j(\T^i_j) \bigg) \\
	& = \sum_{i \in N} \big(\dv_i(M) - \dv_i(\A^*_i)\big) + \sum_{i \in N} \sum_{j \in \A^*_i} \chi_j(\T^i_j),
	\end{align*}
	where the first inequality holds because $\chi_j$ is ${\alpha}$-budget balanced and the second inequality follows from \eqref{eq:local-shares} in the proof of Lemma~\ref{lem:opt-bundle}. 
	
	The proof follows if we can show that 
	\begin{equation}\label{eq:cs-bound}
	\sum_{i \in N} \sum_{j \in \A^*_i} \chi_j(\T^i_j) \le {2}{\alpha^3}H_n \sum_{j \in M} c_j(\T^*_j).
	\end{equation}

	We use a charging argument to prove \eqref{eq:cs-bound}. Fix some item $j \in M$ and order the players in $\T^*_j$ according to the player order $\tau = (1, \dots, n)$ induced by $\M$; let $\T^*_j = \set{i_1, \dots, i_{k^*_j}}$ be the ordered set with $k^*_j := |\T^*_j|$. We now ``tag'' each player $i$ in $\T^*_j$ with a fraction of the cost $c_j(\T^*_j)$ for item $j$ as follows: For the $l$th player $i = i_l$ in $\T^*_j$ with $1 \le l \le k^*_j$, define
	\begin{equation}\label{eq:tag}
	\mtag_i(j) := \frac{c_j(\T^*_j)}{k^*_j-l+1}.
	\end{equation}
	That is, the first player $i_1$ in $\T^*_j$ is tagged with $c_j(\T^*_j)/k^*_j$, the second player $i_2$ with $c_j(\T^*_j)/(k^*_j-1)$ and so forth, and the last player $i_{k_j^*}$ is tagged with $c_j(\T^*_j)$. 
	
	We first derive two lower bounds on the tagged cost:
	\begin{claim}\label{claim:tag-bound}
		For every player $i \in N$ and for every item $j \in \A^*_i$: 
		$$
		\mtag_i(j) \ge \frac{c_j(\T^*_j)}{n-i+1} \qquad\text{and}\qquad \mtag_i(j) \ge \frac{c_j(\T^*_j)}{|T^*_j|}.
		$$
	\end{claim}
	\begin{proof}
		The latter bound holds by definition \eqref{eq:tag}. 
		To see that the former bound holds, observe that the $k$th last player ($1 \le k \le k^*_j$) in the ordered set $\T^*_j$ is tagged by $c_j(\T^*_j)/k$. The claim now follows because there are at most $n-i$ players succeeding $i$ in $\T^*_j$ according to the order. 	
	\end{proof}
	
	Note that the total tagged cost of item $j$ satisfies
	\begin{equation}\label{eq:total-tag}
	\sum_{i \in \T^*_j} \mtag_i(j) = \sum_{l = 1}^{k^*_j} \frac{c_j(\T^*_j)}{k^*_j-l+1} \le H_n c_j(\T^*_j).
	\end{equation}
	
	Thus, to prove \eqref{eq:cs-bound} it suffices to show that the total cost share sum is upper bounded by the total tagged cost, i.e., 
	\begin{equation}\label{eq:acs-bound2}
	\sum_{i \in N} \sum_{j \in \A^*_i} \chi_j(\T^i_j) \le {2}{\alpha^3}\sum_{j \in M} \sum_{i \in \T^*_j} \mtag_i(j). 
	\end{equation}
	We show that for every $i$ and every $j \in \A^*_i$, $\chi_j(\T^i_j) \le \mtag_i(j)$. 
	Summing over all $i \in N$ and $j \in \A^*_i$ then proves \eqref{eq:cs-bound}. 
	
	We distinguish two cases: 
	
	Case 1: $|\T^*_j| \ge n-i+1$: Let $S \subseteq \T^*_j$ be a set such that $|S| = n-i+1$. We have 
	\begin{equation}\label{eq:asc-c1}
	\chi_j(\T^i_j) \leq {\alpha} \frac{c_j(\set{i, \dots, n})}{n-i+1} 
	\leq {2}{\alpha^2}\frac{c_j(S)}{|S|} 
	\leq {2}{\alpha^2}\frac{c_j(\T^*_j)}{n-i+1} 
	\leq {2}{\alpha^2}\mtag_i(j) ,
	\end{equation}
	where  the first inequality follows from Lemma~\ref{lem:acost-share-bound}, the second inequality follows from Lemma~\ref{lem:adiffer-2}, the third inequality holds because $c_j$ is non-decreasing and the last inequality follows from Claim~\ref{claim:tag-bound}.

	Case 2: $|\T^*_j| < n-i+1$:
	Let $S \supset \T^*_j$ be a set such that $|S| = n-i+1$. 
	We have
	\begin{equation}\label{eq:asc-c1}
	\chi_j(\T^i_j) 
	\leq {\alpha} \frac{c_j(\set{i, \dots, n})}{n-i+1} 
	\leq {2}{\alpha^2}\frac{c_j(S)}{|S|} 
	\leq {2}{\alpha^3}\frac{c_j(\T^*_j)}{|\T^*_j|} 
	\leq {2}{\alpha^3}\mtag_i(j) ,
	\end{equation}
	where  the first inequality follows from Lemma~\ref{lem:acost-share-bound}, the second inequality follows from Lemma~\ref{lem:adiffer-2}, the third inequality holds because $c_j$ is $\alpha$-average-decreasing and the last inequality follows from Claim~\ref{claim:tag-bound}. 
	This concludes the proof.
\end{proof}

\section{Mechanisms for General Valuations and Subadditive Cost Functions}\label{sec:FPM}

In this section, we move away from symmetric submodular valuation functions and derive results for more general functions. In particular, we investigate the performance of the {\em Sequential Mechanism} \cite{Moulin99} for general valuations and subadditive cost functions. Although for arbitrary subadditive cost functions this mechanism does not provide favorable approximation guarantees, we identify conditions on the cost functions under which it achieves significantly better approximation factors. This is based on considering a different parameterization of cost functions with regard to their average cost shares. 

\subsection{The Sequential Mechanism}\label{subsec:Seq}

The \emph{Sequential Mechanism (SM)} was introduced by Moulin \cite{Moulin99} and was also studied in \cite{DobzinskiO17}. 
A description of the mechanism SM is given in Algorithm~\ref{fig:alg-3}. We note that this mechanism is applicable both to separable and non-separable cost functions. Here, we first focus on separable cost functions. In Section \ref{subsec:nonsep}, we consider generalizations to the non-separable setting.

\begin{algorithm}[t]
	\small
	\DontPrintSemicolon 
	\caption{Sequential Mechanism (SM)} \label{fig:alg-3} 
	
	\KwIn{Declared valuation functions $(b_i)_{i \in N}$.}
	\KwOut{Allocation $A = (A_i)_{i \in N}$ and payments $p = (p_i)_{i \in N}$.}
	
	\textbf{Initialization:} Fix an order on the set of players $N = \set{1, \dots, n}$. 	
	
	\For{$i = 1, \dots, n$}{
		
		Compute an \emph{optimal bundle} $A_i$ for player $i$:
		\vspace*{-1ex}
		\begin{align*}\label{eq:opt-choice}
		A_i & \in \arg\max_{S \subseteq M} \{b_i(S) - p_i(S)\}, \quad \text{where} \\&\
		\qquad p_i(S) = C(A_1, \dots, A_{i-1}, S, \emptyset, \dots, \emptyset) - C(A_1, \dots, A_{i-1}, \emptyset, \dots, \emptyset). 
		\end{align*}			
		(If there are multiple optimal bundles, choose the lexicographically smallest one.) \;	
	}	
	
	\Return $A = (A_i)_{i \in N}$ and $p = (p_i)_{i \in N}$, where $p_i = p_i(A_i)$.
	
\end{algorithm}

It is trivial to see that SM is budget-balanced and it is also known that it is WGSP \cite{DobzinskiO17}. However, for arbitrary monotone subadditive cost functions, the mechanism achieves a (poor) social cost approximation guarantee of $n$ only (see \cite{DobzinskiO17}). 
\begin{theorem}\cite{DobzinskiO17} 
	Suppose we have general valuation functions and non-decreasing subadditive cost functions. Then the Sequential Mechanism satisfies IR, NPT, WGSP, and is budget balanced and $n$-approximate.
\end{theorem}

Despite this, we show that SM has better guarantees under certain conditions. Namely, we identify a crucial parameter of each cost function $c_j$ with $j \in M$ that allows us to quantify this improvement. The parameterization introduced here is different from the one used in Section \ref{sec:mech} and it compares the average cost of a set $T \subseteq N$ with the minimum standalone cost of a player in $T$. More specifically, we define the following property: 

\begin{definition}
	\label{def:cmin}
	A cost function $c: 2^N \rightarrow \mathbb{R}_{\ge 0}$ is \emph{$\alpha$-average min-bounded}, if there exists some $\alpha \ge 1$ such that for every set $T \subseteq N$, we have $\alpha\cdot \frac{c(T)}{|T|} \geq c_{\min}$, where $c_{\min} = \min_{j\in T}c(\{j\})$.
\end{definition}





Definition \ref{def:cmin} may look somewhat contrived at first glance and we thus provide some more intuition on how we arrived at this parameterization. 
Given that IACSM performs well for $\alpha$-average decreasing functions and small values of $\alpha$ (see Section~\ref{sec:mech}), it is natural to focus on the complement of this class. For example, fix $\alpha = 1$ for now. Then the exact complement is not easy to characterize because it involves two existential quantifiers. We therefore consider a subset of this complement (with only one existential quantifier) by demanding that for every $T$, there exists $S \subseteq T$ such that $c(S)/|S| < c(T)/|T|$. It is not hard to verify that this definition is equivalent to the class of 1-average min-bounded functions. For larger values of $\alpha$, we can again see that $\alpha$-average-min-bounded functions capture a chunk of the complement of $\alpha$-average-decreasing functions. Thus, a positive result for $\alpha$-average-min-bounded functions narrows down on the cost functions that are not yet known to admit good approximation guarantees. 

Note that for every cost function we can find an $\alpha \ge 1$ such that it is $\alpha$-average min-bounded. As the next theorem reveals, the Sequential Mechanism attains a favorable performance for small values of $\alpha$.

\begin{theorem}\label{thm:ultra} 
	Suppose we have general valuation functions and for each item $j\in M$, the cost function $c_j:2^N \rightarrow\mathbb{R}_{\geq 0}$ is non-decreasing, subadditive, and $\alpha$-average min-bounded for some $\alpha\geq 1$. Then the Sequential Mechanism satisfies IR, NPT, WGSP, and is budget balanced and $\alpha\cdot H_n$-approximate.
\end{theorem}

For the proof of Theorem \ref{thm:ultra}, we use the following proposition:

\begin{proposition}\label{prop:betterg}
	If $c: 2^N \rightarrow \mathbb{R}_{\ge 0}$ is non-decreasing and $\alpha$-average min-bounded, then $\sum_{i \in T}c(\{i\}) \leq \alpha H_{|T|} \cdot c(T)$ for every $T \subseteq N$. 
\end{proposition}

\begin{proof}
Fix $T\subseteq N$ and rename the players of $T$ so that they are in decreasing order with respect  to the standalone cost, i.e., for any $i,j \in T$ with $ i<j$, it holds $c(\{i\}) \geq c(\{j\})$. For convenience, we may assume that $T = \{1, 2, \dots, |T|\}$. Fix some $i \in T$ and consider the set $\{1,2,\dots,i\} \subseteq T$. Note that in this set player $i$ has the minimum standalone cost $c(\{i\})$. Hence, by Definition \ref{def:cmin}, the cost function satisfies: 
	$$ 
	c(\{i\})\leq \alpha\cdot \frac{c(\{1,\dots,i\})}{i} \leq \alpha\cdot \frac{c(T)}{i},
	$$
	where the last inequality holds because $c$ is non-decreasing.
	The above inequality holds for any $i \in T$. Summing over all players $i\in T$ proves the claim. 	
\end{proof}

\begin{proof}[Proof of Theorem \ref{thm:ultra}]
	We only need to prove that SM is $\alpha H_n$-approximate. All the other properties have been established in  \cite{DobzinskiO17,Moulin99}.
	Let $\A=(\A_i)_{i \in N}$ be the allocation output by the mechanism and let $\A^* = (\A^*_i)_{i \in N}$ be an optimal allocation. Further, let $\T^*_j$ be the respective set of players that receive item $j$ in $\A^*$. 
	To simplify notation in the analysis, we also let $A_{<i}$ denote the tuple $(A_1, \dots, A_{i-1}, \emptyset, \dots, \emptyset)$. Define now the incremental cost of a player $i$ for a bundle $S\subseteq M$, with respect to the allocation constructed by the Sequential Mechanism before $i$'s turn as $\Delta_i(A_{<i}, S) = C(A_1, \dots, A_{i-1}, S, \emptyset, \dots, \emptyset) - C(A_1, \dots, A_{i-1}, \emptyset, \dots, \emptyset)$.
	
	We have 
	\begin{align*}
	\pi(\A) 
	& = \sum_{i \in N} \big[\dv_i(M) - \dv_i(\A_i)\big] + C(A)\\
	&= \sum_{i \in N} \dv_i(M) - \sum_{i \in N}\big[ \dv_i(\A_i) - \Delta_i(A_{<i}, A_i) \big] \\
	& \le \sum_{i \in N} \dv_i(M) - \sum_{i \in N}\big[ \dv_i(\A^*_i) - \Delta_i(A_{<i}, A^*_i) \big] \\
	&= \sum_{i \in N} \big[\dv_i(M) - \dv_i(\A^*_i)\big] + \sum_{i \in N} \Delta_i(A_{<i}, A^*_i).
	\end{align*}
	Note that the inequality holds because $A_i$ was chosen as the optimal bundle for $i$. The next step is to prove a bound on the incremental costs in the form 
	\begin{equation}
	\label{eq:icb}
	\sum_{i \in N} \Delta_i(A_{<i}, A^*_i) \le \beta \cdot C(A^*).
	\end{equation}
	The proof follows if we can show that \eqref{eq:icb} holds for $\beta = \alpha H_n$ because we then have
	$$ \pi(\A) \le \sum_{i \in N} \big[\dv_i(M) - \dv_i(\A^*_i)\big] + \alpha\cdot H_n C(A^*) \le \alpha H_n\cdot \pi(A^*). $$
	
	By exploiting the subadditivity of the cost functions $c_j$, we obtain 
	\begin{align*}
	\Delta_i(A_{<i}, A^*_i)
	& = C(A_{<i}, A^*_{i}) - C(A_{<i}) 
	\le C(A_{<i}) + C(A^*_{i}, \emptyset_{-i}) - C(A_{<i}) 
	= \sum_{j\in A^*_i} c_j(\set{i}).
	\end{align*}
	
	Summing over all $i\in N$, and using Proposition \ref{prop:betterg}, we get:
	$$\sum_{i \in N} \Delta_i(A_{<i}, A^*_i) \le \sum_{i\in N} \sum_{j\in A^*_i} c_j(\set{i}) = \sum_{j\in M} \sum_{i\in T^*_j} c_j(\set{i}) \le \sum_{j\in M} \alpha H_{|T_j^*|} c_j(T_j^*) \le \alpha H_n C(A^*)$$

	\vneg\vneg
\end{proof}

By going through the proof of Theorem \ref{thm:ultra} more carefully, we realize the following:

\medskip
\begin{remark}
	\label{rem:gen}
	For any subclass of non-decreasing, subadditive cost functions, it suffices to establish inequality \eqref{eq:icb} to prove that the Sequential Mechanism has a social cost approximation guarantee of $\beta$.
\end{remark}

We next prove that for $\alpha = 1$ the approximation factor is tight.

\begin{proposition}\label{prop:tight}
	Even for the single item setting, there exists a $1$-average min-bounded cost function, under which the Sequential Mechanism provides an $H_n$-approximation.
\end{proposition}

\begin{proof}
Consider a set $N=\{p_1, p_2, ..., p_n\}$ of players and the following function $c:2^n \rightarrow\mathbb{R}_{\geq 0}$:
	\begin{align}
	c(S) =  \begin{cases}
	0 & \text{if $S=\emptyset$} \\
	\frac{k}{j} & \text{if $S=\{p_j\}$} \\
	\min\{k, \sum_{p_j \in S}c(\{p_j\})\} & \text{if $|S|\geq 2$} 
	\end{cases} 
	\end{align}
	where $k \geq 0$ and $S \subseteq N$. We have to show that the above function is subadditive, non-decreasing, it has bounded average costs, and that the \emph{Sequential Mechanism} is $H_n$-approximate.

\medskip\noindent
	\textbf{Subadditivity:} Initially notice that if $p_1 \in S \subseteq N$, then $c(S)=k$. Now consider the non-empty sets $A,B \subseteq N$.
	\begin{itemize}\itemsep0pt
		\item If $p_1 \in A \cup B$, then $p_1$ is in at least one of $A,B$. Thus we have that $c(A\cup B) \leq c(A)+c(B)$.
		\item If $p_1 \notin A \cup B$, then
		\begin{align*}
		c(A\cup B) 
		&=\min\{k, \sum_{p_j \in A\cup B}c(\{p_j\})\}\\
		&\leq
		\min\{k, \sum_{p_j \in A}c(\{p_j\})+ \sum_{p_j \in B}c(\{p_j\})\}\\
		&\leq \min\{k, \sum_{p_j \in A}c(\{p_j\})\}+ \min\{k, \sum_{p_j \in B}c(\{p_j\})\}\\
		&=c(A)+c(B)
		\end{align*}
	\end{itemize}

\medskip\noindent	
	\textbf{Non-decreasingness:} Consider the sets $A\subseteq B \subseteq N$.
	\begin{itemize}\itemsep0pt
		\item If $p_1 \in A \Rightarrow p_1 \in B \Rightarrow c(A)=k=c(B)$. 
		\item If $p_1 \notin A$ and $p_1 \in B$, then $C(A)=\min\{k, \sum_{p_j \in A}c(\{p_j\})\} \leq k=c(B)$.
		\item If $p_1 \notin A$ and $p_1 \notin B$, then $c(A)=\min\{k, \sum_{p_j \in A}c(\{p_j\})\} \leq \min\{k, \sum_{p_j \in B}c(\{p_j\})\}=c(B)$.
	\end{itemize}
	
\medskip\noindent
\sloppy
	\textbf{Min-bounded average costs:} 
	Consider the non-empty set of indices $A \subseteq [n]$ and the set $B=\{p_1, p_2, ...,p_{|A|}\}$.
	\begin{itemize}\itemsep0pt
		\item If $p_1 \in A \Rightarrow c(A)=k$. Now notice that $\frac{c(A)}{|A|}=\frac{k}{|A|}=\min_{j \in B}c(\{i_j\}) \geq \min_{j \in A}c(\{i_j\})$. 
		\item If $p_1 \notin A$, then we have that either $c(A)=k$ and thus $\frac{c(A)}{|A|} \geq \min_{j \in A}c(\{i_j\})$ as before, or $c(A)=\sum_{p_j \in A}c(\{p_j\})$ and thus, $\frac{c(A)}{|A|}=\frac{\sum_{p_j \in A}c(\{p_j\})}{|A|} \geq \min_{j \in A}c(\{i_j\})$.
	\end{itemize}

\fussy
\medskip\noindent
	\textbf{Approximation of the Sequential Mechanism:} Consider an instance where the value each player $p_j$ has for the item is  $v_{p_j}=\frac{k}{j}-\epsilon$, for arbitrary small $\epsilon>0$. Suppose now that $A^*$ is the optimal allocation, $A'$ is the allocation where every player is served, and $A$ is the output of the Sequential Mechanism, in which it is easy to see that no player is served. We have that $\pi(A^*) \leq \pi(A')=k$, while 
\begin{align*}
\pi(A)=\ 0+\sum_{p_j \in N}\frac{k}{j}-n\epsilon=kH_n-n\epsilon.
\end{align*}
So since $\epsilon$ is arbitrary small, the approximation cannot be better than $H_n$.
\end{proof}

\subsection{Improved Approximation Guarantees and Applications}\label{subsec:App}


We continue with a natural refinement of Definition \ref{def:cmin} which turns out to provide even better approximation factors of the Sequential Mechanism.

\begin{definition} \label{def:cmax}
	A cost function $c: 2^N \rightarrow \mathbb{R}_{\ge 0}$ is \emph{$\alpha$-average max-bounded}, if there exists some $\alpha \ge 1$ such that for every set $T \subseteq N$, we have $\alpha \cdot \frac{c(T)}{|T|} \geq c_{\max}$, where $c_{\max} = \max_{j\in T}c(\{j\})$.
\end{definition}

Clearly, any function that is $\alpha$-average max-bounded is also $\alpha$-average min-bounded. 
Thus, we already have an $\alpha H_n$-approximation for non-decreasing, subadditive and $\alpha$-average max-bounded cost functions. Below we show that we can achieve a much better guarantee. 

\begin{theorem}\label{thm:gmax} 
	Suppose we have general valuation functions and for each item $j\in M$, the cost function $c_j:2^N \rightarrow\mathbb{R}_{\geq 0}$ is non-decreasing, subadditive, and $\alpha$-average max-bounded for some $\alpha \geq 1$. Then the Sequential Mechanism satisfies IR, NPT, WGSP, and is budget-balanced and $\alpha$-approximate.
\end{theorem}

\begin{proof}
We only need to prove that the mechanism is $\alpha$-approximate. 
	By following exactly the same reasoning as in the proof of Theorem \ref{thm:ultra} and using the observation made in Remark \ref{rem:gen} (note that $c_j$ is subadditive by assumption), we only need to prove that \eqref{eq:icb} holds for $\beta = \alpha$. 
	Exploiting the properties of the cost functions, we obtain
	$$\sum_{i \in N} \Delta_i(A_{<i}, A^*_i) \le \sum_{i\in N} \sum_{j\in A^*_i} c_j(\set{i}) \le \sum_{j\in M} \sum_{i\in T^*_j} \alpha \frac{c_j(T_j^*)}{|T_j^*|} = \alpha \sum_{j\in M} c_j(T_j^*) \le \alpha C(A^*).
	$$

	\vneg\vneg
\end{proof}

\vneg

\myheader{Example applications of combinatorial cost functions}
We give some examples of combinatorial cost functions below and show that they are $\alpha$-average max-bounded (possibly depending on some parameters of the combinatorial problem). In particular, by applying Theorem~\ref{thm:gmax} we obtain attractive social cost approximation guarantees for these problems. 
For simplicity, all examples consider a single item only; but clearly, we can consider more general multiple item settings (e.g., when for each item $j \in M$, $c_j$ captures one of the problems below).
\begin{enumerate}
	\item \textbf{Set Cover.} We are given a universe of elements $U$ and a family $\mathcal{F} \subseteq 2^U$ of subsets of $U$. The players correspond to the elements of $U$ and the cost $c(S)$ for serving a set of players $S \subseteq U$ is defined as the size of a minimum cardinality set cover for $S$.
	\item \textbf{Vertex Cover.} This is a special case of Set Cover. We are given an undirected and unweighted graph $G = (V, E)$ and the players are the edges of the graph. The cost $c(S)$ for serving a set $S \subseteq E$ of players is defined as the size of a minimum vertex cover in the subgraph induced by $S$. 
	\item \textbf{Matching.} We are given an undirected and unweighted graph $G = (V, E)$ and the players correspond to the edges. The cost $c(S)$ for serving a set $S$ of players is defined as the size of a maximum cardinality matching in the subgraph induced by $S$.
\end{enumerate}

Using our $\alpha$-average max-bounded notion, it is now easy to prove that these problems admit constant social cost approximation guarantees (under certain restrictions). 

\begin{theorem}\label{thm:appl}
	The Sequential Mechanism is $\alpha$-approximate for the above problems, where 
	\begin{compactenum}
		\item $\alpha = d$ for the Set Cover problem, where $d$ is the maximum cardinality of the sets in $\mathcal{F}$; 
		\item $\alpha = k$ for the Vertex Cover problem in graphs of maximum degree $k$;
		\item $\alpha = k$ for the Matching problem in bipartite graphs of maximum degree $k$;
		\item $\alpha = (5k+3)/4$ for the Matching problem in general graphs of maximum degree $k$.
	\end{compactenum}
\end{theorem}

\begin{proof} We have the following:
\begin{enumerate}
		\item For Set Cover, it is trivial to check that the cost function is subadditive. Consider a subset $T$ of the players. Note that for a single player $i\in T$, $c(\{i\}) = 1$, since each element can be covered by a single subset. Given that the maximum cardinality of a subset is $d$, the minimum set cover for covering the set of players $T$ is at least $|T|/d$. We conclude that $c$ is $d$-average max-bounded. 
		\item Vertex Cover is a special case of Set Cover. Consider a subset $T\subseteq E$ of the edges. Given that the maximum degree is $k$, the minimum vertex cover for covering a set of edges $T$ is at least $|T|/k$. We conclude that this cost function is $k$-average max-bounded. 
		\item The cost function in the Maximum Matching problem is also subadditive. In bipartite graphs, we also know that the cardinality of a maximum matching equals the cardinality of a minimum vertex cover. This immediately implies a $k$-approximation.
		\item For general graphs, note again that $c(\{i\}) = 1$ for a single player $i\in E$. By the work of \cite{Han08}, we know that in a graph of $m$ edges and with degree $k$, there always exists a matching of size at least $4m/(5k+3)$. Hence for a set of players $T\subseteq E$, we have that $c(T)/|T| \geq 4/(5k+3)$. We conclude that the cost function $c$ is $(5k+3)/4$-average max-bounded. 
	\end{enumerate}
\end{proof}

We compare these bounds with the existing results in the literature: 
For Vertex Cover, there is a mechanism that is $2$-budget-balanced and $O(\log{n})$-approximate \cite{MehtaRS09}. Thus, for graphs with maximum degree less than $\log{n}$, we obtain a better guarantee. 
For Set Cover, there is a mechanisms that is $O(\log{n})$-budget-balanced and $O(\log{n})$-approximate \cite{MehtaRS09}. Hence, we obtain an improvement if the sets in $\mathcal{F}$ have small size. 
Finally, we note that our results do not apply to the weighted versions of these problems.

\subsection{Guarantees of the Sequential Mechanism for Non-Separable Cost Functions}\label{subsec:nonsep}

We extend our results to non-separable cost functions. 
Recall that in this setting, the cost $C(A)$ of an allocation $A = (A_i)_{i \in N}$ is given by some general (not necessarily separable) cost function $C:(2^M)^n \rightarrow \mathbb{R}_{\geq 0}$. In particular, $C$ may encode dependencies among different items. 

We introduce some more notation. Given two allocations $S = (S_i)_{i \in N}$ and $T = (T_i)_{i \in N}$, we define $S \cup T$ as the componentwise union of $S$ and $T$, i.e., $S \cup T = (S_1 \cup T_1, \dots, S_n \cup T_n)$. Similarly, we write $S \subseteq T$ if this relation holds componentwise, i.e., $S_i \subseteq T_i$ for every $i \in N$. 
Given an allocation $A = (A_i)_{i \in N}$ and a set of players $S \subseteq N$, we define $A|_{S} = (A_S, \emptyset_{-S})$ as the allocation in which each player $i \in S$ receives the items in $A_i$ and all other players receive nothing.
If $S = \set{i}$ is a singleton set, we also write $A|_i$ instead of $A|_{\set{i}}$.
Throughout this section, we remain in the domain of non-decreasing and subadditive cost functions. In the non-separable case, a cost function $C:(2^M)^n \rightarrow \mathbb{R}_{\geq 0}$ is non-decreasing if $C(S) \le C(T)$ for every pair of allocations $S, T$, with $S \subseteq T$. Also, it is subadditive if for every two allocations $S = (S_i)_{i \in N}$ and $T = (T_i)_{i \in N}$, we have $C(S \cup T) \leq C(S)+C(T)$.


We now adapt Definitions \ref{def:cmin} and \ref{def:cmax} to non-separable cost functions.

\begin{definition}
	Let $C:(2^M)^n \rightarrow \mathbb{R}_{\geq 0}$ be a non-separable cost function.
	\begin{compactitem}\itemsep0pt
		\item $C$ is \emph{$\alpha$-average min-bounded}, if there exists some $\alpha\geq 1$ such that for every allocation $A$ and every subset $T \subseteq N$ with $|T| \ge 2$, it holds $\alpha \frac{C(A|_T)}{|T|} \geq C_{\min}$, where $C_{\min} = \min_{j\in T} C(A|_j)$. 
		
		\item $C$ is \emph{$\alpha$-average max-bounded}, if there exists some $\alpha\geq 1$ such that for every allocation $A$ and every subset $T \subseteq N$ with $|T| \ge 2$, it holds $\alpha \frac{C(A|_T)}{|T|} \geq C_{\max}$, where $C_{\max} = \max_{j\in T} C(A|_j)$.
	\end{compactitem}
\end{definition}


As before, if a non-separable function is $\alpha$-average max-bounded, then it is also $\alpha$-average min-bounded. 

We remark that it has been shown in \cite{Moulin99,DobzinskiO17} that the Sequential Mechanism is weakly group-strategyproof and budget balanced for the non-separable setting. We proceed by adapting Proposition~\ref{prop:betterg} for the non-separable setting.

\begin{proposition}\label{prop:genbetterg}
	If $C:(2^M)^n \rightarrow \mathbb{R}_{\geq 0}$ is a non-separable cost function, which is non-decreasing and $\alpha$-average min-bounded, then for every allocation $A$, $\sum_{i \in T} C(A|_i) \leq H_{|T|} \cdot C(A|_T)$ for every $T \subseteq N$.
\end{proposition}

\begin{proof}
	Let $A$ be an allocation and fix $T\subseteq N$. 
	Rename the players in $T$ such that for any $i, j \in T$ with $i < j$, it holds that $C(A|_i) \ge C(A|_j)$. For convenience, we may assume that $T = \{1, 2, \dots, |T|\}$. Fix some $i \in T$ and consider the set $\set{1, 2, \dots, i} = [i]\subseteq T$. By exploiting that $C$ is $\alpha$-average min-bounded and non-decreasing, we have
	$$
	C(A|_i) \leq \alpha \frac{C(A|_{[i]})}{i} \leq \frac{C(A|_T)}{i}.	$$ 
	
	Note that the above inequality holds for any $i \in T$. 
	Summing over all players $i\in T$ proves the claim. 	
\end{proof}

Now using the same reasoning as in the proof of Theorem \ref{thm:ultra}, we obtain the same approximation guarantee of $\alpha H_n$ as in the separable setting. Further, the improvement we obtained in Theorem \ref{thm:gmax} also goes through in this setting. We summarize these observations in the following corollary. 

\begin{corollary}\label{col:seqgen} 
	Suppose we have general valuation functions and a non-decreasing, subadditive, and $\alpha$-average min-bounded cost function $C:(2^M)^n \rightarrow\mathbb{R}_{\geq 0}$. Then the Sequential Mechanism satisfies IR, NPT, WGSP, and is budget balanced and $\alpha \cdot H_n$-approximate. Furthermore, if $C$ is also $\alpha$-average max-bounded, then the Sequential Mechanism is $\alpha$-approximate.
\end{corollary}

\section{Discussion}\label{sec:dis}

In Section \ref{sec:mech}, we proposed the mechanism IACSM, which is weakly group-strategyproof under general cost functions and symmetric submodular valuations. Moreover it is $\alpha$-budget balanced and $2\alpha^3 H_n$-approximate when we restrict the cost functions to the {$\alpha$-average-decreasing} class. The social cost approximation guarantee further improves to $H_n$ if the cost functions are symmetric submodular and this is best possible (due to the known lower bound for public-excludable goods \cite{DobzinskiMRS18}). It would be very interesting to explore mechanisms that go beyond symmetric submodular valuation functions. It seems that entirely new ideas are needed for this setting. It would also be interesting to extend our mechanism to non-separable cost functions. We note that separability of the costs in Section \ref{sec:mech} is needed for IACSM only to argue that the cost share per item increases as players withdraw (with respect to the trace). One would need to investigate how to adapt the mechanism and enforce this property in the non-separable setting. Technically, this seems far from obvious and we leave a proper treatment of this issue for future work.

In Section \ref{sec:FPM}, we studied the (partially) complementary class of {$\alpha$-average min bounded} cost functions. We showed that the well-known {Sequential Mechanism} is budget balanced and $\alpha H_n$-approximate even for general valuation functions. These results also extend to non-separable cost functions. A very natural question is whether SM is optimal in this setting and we note that the answer is not yet clear: The impossibility result of \cite{DobzinskiMRS18} holds for the public-excludable good cost function which is symmetric submodular and thus {1-average-decreasing}. However, it is not hard to see that this does not fall within the {$\alpha$-average min bounded} class for any constant $\alpha$. This leads to the question of whether there exists a WGSP mechanism that breaks the $\Omega(\log(n))$-approximation in terms of social cost for {$\alpha$-average-min bounded} functions with small values of $\alpha$.

Finally, what we also find very interesting is to identify the class of cost functions for which neither of the two mechanisms studied here perform well. Recall that, for any constant value of $\alpha$, if a cost function is either $\alpha$-average decreasing or $\alpha$-average min-bounded, then a good performance is guaranteed. Thus, we need look at the complement of the set of $\alpha$-average decreasing functions and the set of $\alpha$-average min-bounded functions for small value of $\alpha$ and examine whether these complements have a non-empty intersection. The following proposition shows that this intersection is indeed non-empty.

\begin{proposition}\label{prop:intersec}
	Given $\alpha\geq 1$, the intersection of the complements of $\alpha$-average-decreasing and $\alpha$-average min-bounded functions is non-empty.
\end{proposition}

\begin{proof}
We begin by defining the complements of the two sets:
	\begin{itemize}
		\item \emph{Complement of $\alpha$-average-decreasing}: For every constant $\alpha \geq 1$ that we choose, there exist sets $S\subseteq T \subseteq N$ such that: $\alpha \frac{C(S)}{|S|}<  \frac{C(T)}{|T|}$.
		\item \emph{Complement of $\alpha$-average min-bounded}: For every constant $\alpha \geq 1$ that we choose, there exists a set $T \subseteq N$ such that: $\alpha \frac{C(T)}{|T|}<  C(\{i\})$, where $i = \argmin_{j\in T}C(\{j\})$.
	\end{itemize}

	We will prove this for the case where there exists only a single item. Consider set $N=\{p_1, p_2, ..., p_n\}$ of players and the following function $c:2^n \rightarrow\mathbb{R}_{\geq 0}$:
	\begin{align}
	c(S) = \begin{cases}
	0 & \text{if $S=\emptyset$} \\
	\sqrt{j} & \text{if $S=\{p_j\}$} \\
	\max_{j \in S}c(\{p_j\}) & \text{if $|S|\geq 2$} 
	\end{cases} 
	\end{align}
	where $S \subseteq N$. Before we proceed with the proof, notice that this function is also subadditive and by definition non-decreasing.
	\begin{itemize}
		\item \emph{Complement of $\alpha$-average-decreasing}: Set $S=\{p_2\}$ and $T=\{p_2, p_n\}$. Now notice that $\alpha< \frac{\sqrt{n}}{2} \Rightarrow 2\alpha< \sqrt{n}$, which holds since for any $\alpha \geq 1$ that we choose, we can always find a large enough $n$ for the inequality to be true.
		\item \emph{Complement of $\alpha$-average min-bounded}: $S=\{p_2\}$ and $T=N$. Now notice that $\alpha \frac{\sqrt{n}}{n}<1\Rightarrow \sqrt{n}<\frac{n}{\alpha}$, which holds since for any $\alpha \geq 1$ that we choose, we can always find a large enough $n$ for the inequality to be true.
		
	\end{itemize}
Therefore, the intersection of the complements is non-empty and the proof is complete.
\end{proof}

Notice that the proof of this proposition follows by constructing a cost function that requires non-constant values of $\alpha$ to be captured by either of our parameterizations. Although the intersection turns out to be non-empty, the constructed cost function is rather artificial and more natural examples are elusive so far. In fact, for most of the known cost functions that have been studied in the literature, at least one of our mechanisms achieves an $O(H_n)$-approximation. To make further progress, we believe it is important to understand better the class of functions defined by the intersection of the two complements, as it would help us to identify the missing elements for deriving mechanisms for a wider class of cost functions.

\section{Acknowledgements}

Part of this work was done while Georgios Birmpas was an intern of the Networks and Optimization group at Centrum Wiskunde \& Informatica. Georgios Birmpas was also partially supported by the ERC Advanced Grant 321171 (ALGAME).

\newpage

\bibliographystyle{plainurl}
\bibliography{CSoCFF}

\end{document}